\documentclass{llncs}

\usepackage[T1]{fontenc}
\usepackage[utf8]{inputenc}
\usepackage[USenglish]{babel}
\usepackage{amsmath}
\usepackage{graphicx}
\usepackage{xspace}
\usepackage[nocompress]{cite}
\usepackage{url}

\makeatletter
\spn@wtheorem{observation}{Observation}{\bfseries}{\itshape}
\makeatother

\newif\ifarxiv
\arxivtrue

\ifarxiv
\newcommand\cpmonly[1]{}
\newcommand\arxivonly[1]{#1}
\newcommand\citeappendix{\relax\xspace}
\subtitle{(with appendix)\vspace{-9pt}}
\else
\newcommand\cpmonly[1]{#1}
\newcommand\arxivonly[1]{}
\newcommand\citeappendix{~\cite{mrmappendix}\xspace}
\fi
\newcommand\proofappendix{{\upshape (Proof in appendix\citeappendix.\!)}\xspace}
\newcommand\fullproofappendix{{\upshape (Extended proof in appendix\citeappendix.\!)}\xspace}
\newcommand\detailsappendix{{\upshape (Further details in appendix\citeappendix.\!)}\xspace}

\title{Most Recent Match Queries\penalty-1 in On-Line Suffix Trees}
\author{N.\,Jesper Larsson}

\institute{IT University of Copenhagen, Denmark,
\email{jesl@itu.dk}}

\newcommand\pred{{\mathit{pred}}}
\newcommand\pos{{\mathit{pos}}}

\newcommand\length{{\mathit{slen}}}
\newcommand\down{{\mathit{down}}}
\newcommand\suf{{\mathit{suf}}}
\newcommand\rsuf{{\mathit{rsuf}}}
\newcommand\repr{{\mathit{repr}}}

\newcommand\emptystring{{\epsilon}}
\newcommand\ST{{\mathcal{ST}}}
\newcommand\STi[1]{{\mathcal{ST}_{\kern-.222em #1}}}
\newcommand\LT{{\mathcal{LT}}}
\newcommand\TT{{\mathcal{T}}}
\newcommand\alphabet{{\mathrm{\Sigma}}}
\newcommand\reprupd{{\mathbf{repr\kern-.1em\textit{-}update}}}
\newcommand\mrmfind{{\mathbf{mrm\kern-.1em\textit{-}find}}}
\newcommand\ltdepth{{\mathit{depth}_\LT}}
\let\bibtla\textsc

\DeclareMathSymbol{\stabove}{\mathord}{symbols}{"3F}
\DeclareMathSymbol{\lttop}{\mathord}{symbols}{"60}

\begin{document}

\maketitle

\begin{abstract}
  A suffix tree is able to efficiently locate a pattern in an indexed string,
  but not in general the most recent copy of the pattern in an online stream,
  which is desirable in some applications. We study the most general version of
  the problem of locating a most recent match: supporting queries for arbitrary
  patterns, at each step of processing an online stream. We present
  augmentations to Ukkonen's suffix tree construction algorithm for
  optimal-time queries, maintaining indexing time within a logarithmic factor in the
  size of the indexed string. We show that the algorithm is applicable to
  sliding-window indexing, and sketch a possible optimization for use in the
  special case of Lempel-Ziv compression.
\end{abstract}

\arxivonly{\thispagestyle{plain}}

\section{Introduction}\label{sec-intro}

The \emph{suffix tree} is a well-known data structure which can be used for
effectively and efficiently capturing patterns of a string, with a variety of
applications~\cite{Apostolico85,gusfield,LarssonPhD}. Introduced by
Weiner~\cite{Weiner73}, it reached wider use with the construction algorithm of
McCreight~\cite{McR}. Ukkonen's algorithm~\cite{UkkoOnli} resembles
McCreight's, but has the advantage of being fully \emph{online}, an important
property in our work.  Farach~\cite{FarFOCS} introduced recursive suffix
tree construction, achieving the same asymptotic time bound as sorting the
characters of the string (an advantage for large alphabets), but at the cost
inherently off-line construction. The simpler \emph{suffix array} data
structure~\cite{Manber93,puglisi2007taxonomy} can replace a suffix tree in many
applications, but cannot generally provide the same time complexity, e.g., for
online applications.

Arguably the most basic capability of the suffix tree is to efficiently locate
a string position matching an arbitrary given pattern. In this work, we are
concerned with finding the \emph{most recent} (rightmost) position of the
match, which is not supported by standard suffix trees. A number of authors
have studied special cases of this problem, showing
applications in data compression and
surveillance~\cite{ALU_timestamp,ferragina1,CrochemoreRightmost}, but to our
knowledge, no efficient algorithm has previously been presented for the general
case. One of the keys to our result is recent advancement in online suffix tree
construction by Breslauer and Italiano~\cite{breslauer_sufext}.

We give algorithms for online support of locating the most recent longest match
of an arbitrary pattern $P$ in $O(|P|)$ time (by traversing $|P|$ nodes, one of
which identifies the most recent position). When a stream consisting of $N$
characters is subject to search, the data structure requires $O(N)$ space, and
maintaining the necessary position-updated properties takes at most $O(N\log
N)$ total indexing time. If only the last $W$ characters are subject to search
(a \emph{sliding window}), space can be reduced to $O(W)$ and time to $O(N\log
W)$.

In related research, Amir, Landau and Ukkonen~\cite{ALU_timestamp}
gave an $O(N\log N)$ time algorithm to support queries for the most
recent previous string matching a suffix of the (growing) indexed
string. The pattern to be located is thus not arbitrary, and the data
structure cannot support sliding window indexing.

A related problem is that of Lempel-Ziv factorization~\cite{LZ77}, where it is
desirable to find the most recent occurrence of each factor, in order to reduce
the number of bits necessary for subsequent encoding. For this special case,
Ferragina et al.~\cite{ferragina1} gave a suffix tree based linear-time
algorithm, but their algorithm is not online, and cannot index a sliding
window. Crochemore et al.~\cite{CrochemoreRightmost} gave an online algorithm
for the rightmost \emph{equal cost} problem, a further specialization for the
same application. In section~\ref{sec-lz}, we discuss a possible optimization
of our algorithm for the special case of Lempel-Ziv factorization.

\section{Definitions and Background}\label{sec-defs}

We study indexing a string $T=t_0\cdots t_{N-1}$ of length $|T|=N$, characters
$t_i\in \alphabet$ drawn from a given alphabet $\alphabet$. (We
consistently denote strings with uppercase letters, and characters with
lowercase letters.)  $T$ is made available as a \emph{stream}, whose
total length may not be known. The index is
maintained online, meaning that after seeing $i$ characters, it is
functional for queries on the string $t_0\cdots t_{i-1}$. Following the
majority of previous work, we assume that $|\alphabet|$ is a
constant.\footnote{It should be noted, however, that ours and previous
  algorithms can provide the same \emph{expected} time bounds for non-constant
  alphabets using hashing, and only a very small worst-case factor higher
  using efficient deterministic dictionary data structures.}

The data structure supports queries for the most recent longest match in
$T$ of arbitrary strings that we refer to as \emph{patterns}. More specifically, given a pattern $P=p_0\cdots p_{|P|-1}$, a
\emph{match} for a length-$M$ prefix of $P$ occurs in position $i$ iff $p_j =
t_{i+j}$ for all $0\leq j < M$. It is a \emph{longest} match iff $M$ is
maximum, and the \emph{most recent} longest match iff $i$ is the maximum
position of a longest match.

\subsection{Suffix Tree Construction and Representation}

By $\ST$, we denote the \emph{suffix tree}~\cite{Weiner73,McR,UkkoOnli,gusfield} 
over the string $T=t_0 \cdots t_{N-1}$. This section defines $\ST$, and
specifies our representation.

A string $S$ is a nonempty \emph{suffix} (of $T$, which is implied) iff
$S=t_i\cdots t_{N-1}$ for $0\le i < N$, and a nonempty \emph{substring} (of
$T$) iff $S=t_i\cdots t_{j}$ for $0\le i\le j < N$. By convention, the empty
string $\emptystring$ is both a suffix and a substring. Edges in $\ST$ are
directed, and each labeled with a string. Each point in the tree, either
coinciding with a node or located between two characters in an edge label,
corresponds to the string obtained by concatenating the edge labels on the path
to that point from the root.  $\ST$ represents, in this way, all substrings
of~$T$. We regard a point that coincides with a node as located at the end of
the node's edge from its parent, and can thus uniquely refer to the point on an
edge of any
represented string. An \emph{external} edge is an edge
whose endpoint is a leaf; other edges are \emph{internal}. The endpoint of each
external edge corresponds to a suffix of $T$, but some suffixes may be
represented inside the tree. Note that the point corresponding to an arbitrary
pattern can be located (or found non-existent) in time proportional to the
length of the pattern, by scanning characters left to right, matching edge
labels from the root down.

We do not require that $T$ ends with a unique
character, which would make each suffix correspond to some edge endpoint. Instead, we maintain points of implicit suffix nodes using the technique
of Breslauer and Italiano~\cite{breslauer_sufext}
(section~\ref{sec-case-suffix}).

Following Ukkonen, we augment the tree with an auxiliary node $\stabove$ above
the root, with a single downward edge to the root. We denote this edge $\lttop$
and label it with $\emptystring$. (Illustration in figure~\ref{fig-st}.)
Although the root of a tree is usually taken to be the topmost node, we
shall refer to the node below $\stabove$ (the root of the unaugmented tree) as
the root node of $\ST$.

Apart from $\lttop$, all edges are labeled with nonempty strings,
and the tree represents exactly the substrings of $T$ in the minimum number of
nodes. This implies that each node is either $\stabove$, the root, a leaf, or a
non-root node with at least two downward edges. Since the number of leaves is
at most $N$ (one for each suffix), the total number of nodes never exceeds
$2N+1$.

We generalize the definition to $\STi{i}$ over the string $T=t_0 \cdots t_{i-1}$,
where $\STi{N}=\ST$. In iteration~$i$, we execute Ukkonen's \emph{update}
algorithm~\cite{UkkoOnli} to reshape $\STi{i-1}$ into $\STi{i}$, without
looking ahead any further than $t_{i-1}$. When there is no risk of ambiguity,
we refer to the current suffix tree simply as $\ST$, implying that $N$
iterations have completed.

For downward tree navigation, we maintain $\down(e, a)=f$ for constant-time
access, where $e$ and $f$ are adjacent edges such that $e$'s endpoint coincides
with $f$'s
start node, and the first character in $f$'s label is $a$. Note that $a$ uniquely
identifies $f$ among its siblings. We define the string that \emph{marks} $f$
as the shortest string represented by $f$ (corresponds to the point
just after $a$). We also maintain $\pred(f)=e$ for constant-time upward
navigation.

For linear storage space in $N$, edge labels are represented indirectly,
as references into $T$. Among the many possibilities for representation, we choose the
following: For any edge $e$, we maintain $\pos(e)$, a position
in $T$ of the string corresponding to $e$'s endpoint, and for each
\emph{internal} edge $e$, we maintain $\length(e)$, the length of that same
string. I.e., $e$ is labeled with $t_i\cdots t_j$, where
$i=\pos(e)+\length(\pred(e))$ and $j=\pos(e)+\length(e)$. External edges need
no explicit $\length$ representation, since their endpoints
always correspond to suffixes of $T$, so $\length(e)$ for external $e$
would always be~$N-\pos(e)$. Note that $\pos(e)$ is not uniquely defined
for internal $e$. Algorithms given in the following sections
update $\pos$ values to allow efficiently finding the most recent occurrence
of a pattern.

Ukkonen's algorithm operates around the \emph{active point}, the point of
the longest suffix that also appears earlier in $T$. This is the deepest point
where $\ST$ may need updating in the next iteration, since longer suffixes are
located on external edges, whose representations do not change. In iteration
$i$, $t_i$ is to be incorporated in $\ST$. If $t_i$ is already present just
below the active point, the tree already contains all the suffixes ending at
$t_i$, and the active point simply moves down past $t_i$. Otherwise, a leaf is
added at the old active point, which is made into a new explicit node if
necessary, and we move to the point of the next shorter
suffix. To make this move efficient, typically jumping to a different branch of
the tree, the algorithm maintains a \emph{suffix link} from any node
corresponding to $aA$, for some character $a$ and string $A$, directly to the
node for~$A$.

\begin{figure}[t]
\vspace{-5mm}
  \begin{center}
\includegraphics[scale=.83]{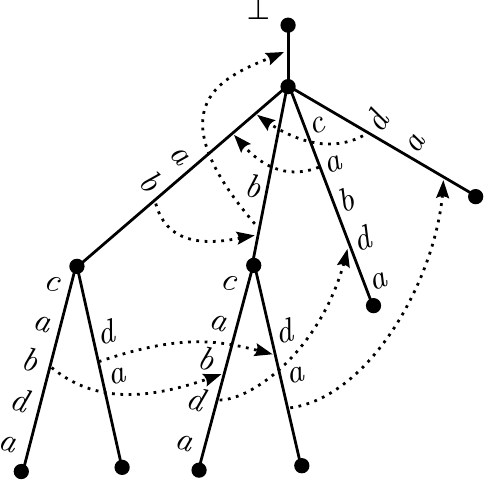}
\end{center}
\vspace{-4mm}
\caption{\label{fig-st}Suffix tree over the string $abcabda$. Dotted
  lines show edge-oriented suffix links.}

\end{figure}
 
We choose a representation where suffix links are edge-oriented, rather than
node-oriented as in McCreight's and Ukkonen's algorithms: 
for edges $e$ and
$f$,
 we let $\suf(e)=f$ iff $A$ marks $f$, and $aA$ and  is the shortest string represented by $e$ such
that $A$ marks an edge. (Illustrated in
figure~\ref{fig-st}.) Furthermore, we define $\rsuf$ to denote the
\emph{reverse suffix link}: $\rsuf(f, a)=e$. We leave $\suf(\lttop)$
undefined. Note that $aA$ is the string that
marks $e$, unless $e$ is a downward edge of the root with an edge label longer
than one character.  We have $\suf(e)=\lttop$ iff $e$'s endpoint corresponds to
a string of length one. 
This variant of suffix links facilitates the description of our \emph{most
  recent match} scheme, but also has practical impact on runtime behavior, due
to reduced branch lookup~\cite{edge_suflink}. The change it implies in
Ukkonen's algorithm is relatively straightforward, and has no impact on its
asymptotic time complexity. We omit the details in this work.

We refer to the path from the active point to $\lttop$, via suffix links and
(possibly) downward edges, as the \emph{active path}. All suffixes that also
appear as substrings elsewhere in $T$ are represented along this path. We refer
to those suffixes as \emph{active
  suffixes}. A key to the $O(N)$ time complexity of Ukkonen's algorithm is that
the active path is traversed only in the forward direction.

\section{Algorithm and Analysis}\label{sec-main}

To answer a most-recent longest-match query for a pattern $P'$, we first locate
the edge $e$ in $\ST$ that represents the longest prefix $P$ of $P'$.  For an
\emph{exact}-match query, we report failure unless $P=P'$. The time required
to locate $e$, by traversing edges from the root, while scanning edge
labels, is $O(|P|)$~\cite{Weiner73,McR,UkkoOnli,gusfield}. In this section, we
give suffix tree augmentations that allow computing the most recent match of
$P$ once its edge is located, while maintaining $O(|P|)$ query
time.

\paragraph{Separation of Cases}

The following identifies two cases in locating the most recent match of a
pattern string $P$, which we treat separately.

\begin{lemma}\label{lem-cases}
  Let $e$ be the edge that represents $P$, and let the string corresponding to
  $e$'s endpoint be $PA$,
    $|A|\ge 0$. Precisely one of the following holds:
  \begin{enumerate}
  \item\label{case-unique} The position of the most recent occurrence of $P$ is also
    the position of the most recent occurrence of $PA$.
  \item\label{case-suffix} There exists a suffix $PB$, $|B|\ge 0$ such that $|B|<|A|$.
  \end{enumerate}
\end{lemma}

\proofappendix Sections \ref{sec-first-case-unique}--\ref{sec-last-case-unique} show how to
deal with
case~\ref{case-unique}, and section~\ref{sec-case-suffix} with case~\ref{case-suffix}.

\subsection{Naive Position Updating}\label{sec-naive}\label{sec-first-case-unique}

We begin with considering a naive method, by which we update $\pos(e)$ at any
time when the string corresponding to $e$'s endpoint reappears in the input.

Observe that any string that occurs later in $t_0\cdots t_{N-1}$ than in
$t_0\cdots t_{N-2}$ must be a suffix $t_j\cdots t_{N-1}$, for some $0\leq j\leq
N-1$. Hence, in each iteration, we need update $\pos(e)$ only if $e$'s endpoint
corresponds to an active suffix. This immediately suggests the following: after update iteration $i$, traverse the active
path, and for any edge $e$ whose endpoint corresponds to a suffix, \emph{pos-update $e$}, which we define as setting
$\pos(e)$ to $i-\length(e)$. Thereby, we maintain $\pos(e)$ as the most recent position for any non-suffix
represented by $e$, and whenever
case~\ref{case-unique} of lemma~\ref{lem-cases} holds, we obtain the most
recent position of $P$ directly from the $\pos$ value of its edge.

The problem with this naive method is that traversing the whole active path in
every iteration results in $\mathrm{\Omega}(N^2)$ worst case time.  The following
sections describe how
to reduce the number of pos-updates, and instead letting
the query operation inspect $|P|$ edges in order to determine the most recent
position.

\subsection{Position Update Strategy}\label{sec-updstrat}


To facilitate our description, we define the \emph{link tree} $\LT$ as the tree
of $\ST$ edges incurred by the suffix links: edges in $\ST$ are nodes in
$\LT$, and $f$ is the parent of $e$ in $\LT$ iff $\suf(e)=f$. The root of $\LT$
is $\lttop$. In order to keep the relationship between $\ST$ edges and
$\LT$ nodes clear, we use the letters $e$, $f$, $g$, and $h$ to denote them in both
contexts.

We define $\ltdepth(e)$ as the depth of $e$ in $\LT$. Because of the
correspondance between $\LT$ nodes and $\ST$ edges, we have
$\ltdepth(e)=\length(\pred(e))$.

By the current \emph{update edge} in iteration $i$, we denote the edge $e$ such
that $\ltdepth(e)$ is maximum among the edges, if any, that would be updated by the
naive update strategy (section~\ref{sec-naive}) in that iteration: the
maximum-$\ltdepth$ internal edge whose endpoint corresponds to an active
suffix. Section~\ref{sec-bresl-updpoint} describes how the update edge can be
located in constant time.

Our update strategy includes pos-updating \emph{only} the update
edge, leaving $\pos$ values corresponding to shorter active suffixes
unchanged. When no update edge exists, we pos-update nothing.
We introduce an
additional value $\repr(e)$ for each internal edge $e$, for which we uphold the
following property:
\begin{property}\label{prop-repr}
  For every node $g$ in the suffix link tree, let $e$ be the most recently pos-updated node in the
  subtree rooted at $g$. Then an ancestor $a$ of $g$ exists such that
  $\repr(a)=e$.
\end{property}

By convention, a tree node is both an ancestor and a descendent of itself.
For new $\LT$ nodes $e$ (without descendants), we set $\repr(e)$ to $\lttop$.
We proceed with first the algorithm that exploits
property~\ref{prop-repr}, then the algorithm to
maintain it.

\subsection{Most Recent Match Algorithm}\label{sec-mrmfind}


Algorithm $\mrmfind(e)$ scans the $\LT$ path from node $e$ to the root in
search for any node $g$ such that $f=\repr(g)$ is a descendent of $e$. For each
such $f$, it obtains the position $q=\pos(f)+\ltdepth(f)-\ltdepth(e)$, and the
value returned from the algorithm is the maximum among the $q$.

\medskip

\noindent $\mrmfind(e)$:\vspace{-1.5ex}
\begin{enumerate}
\item Let $p=\pos(e)$, and $g=e$.
\item\label{st-mrmfind-loop} If $g$ is $\lttop$, we are done, and terminate returning
  the value $p$.
\item If $\repr(g)=\lttop$ (i.e., it has not been set), go directly to step~\ref{st-mrmfind-next}.
\item\label{st-mrmfind-isancestor}  Let $f=\repr(g)$. If $e$ is not an ancestor of $f$ in $\LT$, go directly
  to step~\ref{st-mrmfind-next}.
\item Let $q=\pos(f)+\ltdepth(f)-\ltdepth(e)$. If $q>p$, set $p$ equal to
  $q$.
\item\label{st-mrmfind-next} Set $g$ to $\suf(g)$, and repeat from step~\ref{st-mrmfind-loop}.
\end{enumerate}

The following lemma establishes that when property~\ref{prop-repr} is
maintained, the most recent occurrence of the string corresponding to $e$'s
endpoint is among the positions considered by $\mrmfind(e)$.

\begin{lemma}\label{lem-find-corr}
  For an internal edge $e$, let $A$ be the string corresponding to $e$'s
  endpoint, and $t_{i-|A|}\cdots t_{i-1}$ the most recent occurrence of $A$ in
  $T$. Then $e$ has a descendent $f$ in $\LT$ whose endpoint corresponds to
  $BA$ for some string $B$, and
  $\pos(f)=i-|B|-|A|$. \proofappendix
\end{lemma}

Since $\mrmfind(e)$ returns the maximum among the considered positions, this
establishes its validity for finding the most recent position of the string
corresponding to $e$'s endpoint. Under case~\ref{case-unique} of
lemma~\ref{lem-cases}, this is the most recent position of \emph{any} string
represented by $e$. Hence, given that $e$ represents pattern $P$, $\mrmfind(e)$
produces the most recent position of $P$ in this case.


\begin{lemma}\label{lem-find-time}
  Execution time of $\mrmfind(e)$, where $e$ represents a string $P$ can be
  bounded by $O(|P|)$. \proofappendix
\end{lemma}
%

\subsection{Maintaining Property \ref{prop-repr}}\label{sec-reprupd}\label{sec-last-case-unique}


Since $\lttop$ is an ancestor of all nodes in $\LT$, we can trivially uphold
property~\ref{prop-repr} in relation to any updated node $e$ simply by setting
$\repr(\lttop)=e$. But since this ruins the property in relation to other nodes
(unless the \emph{previous} value of $\repr(\lttop)$ was an
ancestor of $e$) we must recursively push the overwritten $\repr$ value down
$\LT$ to the root of the subtree containing those nodes.

More specifically, when $\repr(r)$ is set to $e$, for some $\LT$
nodes $r$ and $e$, let
$f$ be the previous value of $\repr(r)$. Then find $h$, the minimum-depth node that
is an ancestor of $f$ but not of $e$, and recursively update $\repr(h)$ to
$f$. To find $h$, we first locate $g$, the lowest common ancestor of
$e$ and $f$. Figure~\ref{fig-push} shows the five different ways in which $e$,
$f$, $g$, and $h$ can be located in relation to one another. In case~a, $h$
lies just under the path between $e$ and the root,
implying that we need to set $\repr(h)$ to $f$. We find $h$ via
a reverse suffix link from $g$. Cases~b (where $f=h$) and~c ($g=e$) are merely
special cases of the situation in~a, and are handled in exactly the
same way. In case~d ($g=f$), the overwritten $\repr$ value
points to an ancestor of $e$, and the process can terminate immediately. Case~e
is the special case of~d where the old and new $\repr$
values are the same.

The following details the procedure. It is invoked as
$\reprupd(e, \lttop)$ in order to reestablish property~\ref{prop-repr}, where
$e$ is the current update edge.

\medskip

\noindent $\reprupd(e, r)$:\vspace{-1.5ex}
\begin{enumerate}
\item Let $f$ be the old value of $\repr(r)$, and set its new value to $e$.
\item If $f=\lttop$ (i.e.\@ $\repr(r)$ has not been previously set), then terminate.
\item\label{st-lca} Let $g$ be the lowest common ancestor of $e$ and $f$.
\item If $g=f$, terminate.
\item Let $h=\rsuf(g, t_j)$, where
  $j=\pos(f)+\ltdepth(f)-\ltdepth(g)-1$.
\item Recursively invoke $\reprupd(f,h)$.
\end{enumerate}



\begin{figure}[t]
\vspace{-.5cm}
  \begin{center}
\includegraphics[scale=.83]{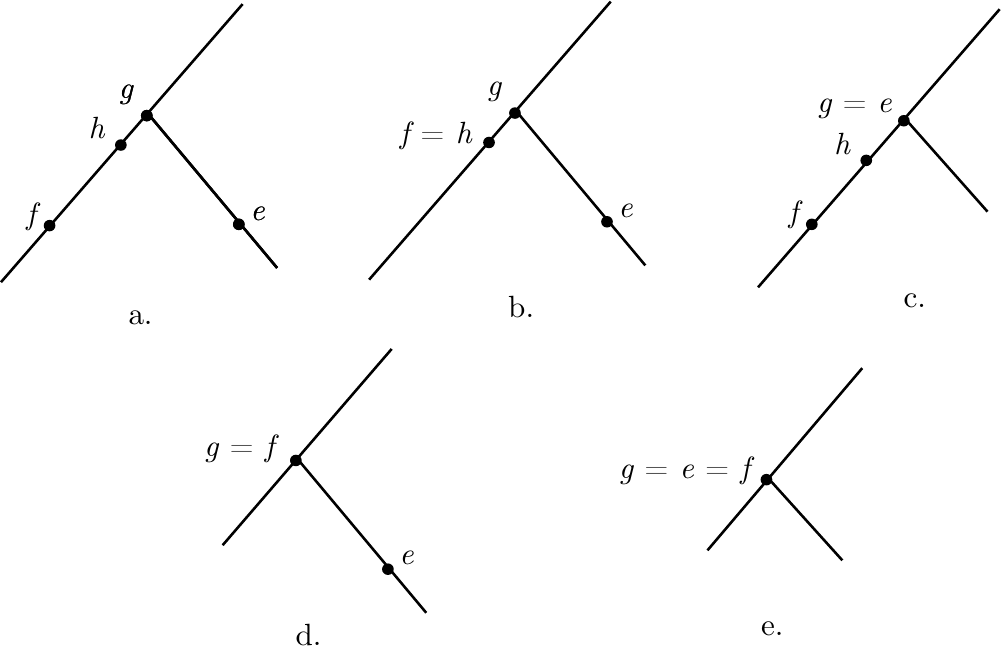}
\end{center}
\vspace{-.3cm}
\caption{\label{fig-push}Cases in $\reprupd$: a, b, and c
  progress down the tree; d and e terminate.}
\end{figure}

Correctness of $\reprupd$ in maintaining
property~\ref{prop-repr}, is established by the preceding discussion.
We now turn to bounding the total
number of recursive calls.

\begin{lemma}\label{lem-balleaves}
  Given a sequence $V=e_1,\ldots, e_N$ of nodes to be updated in a tree $\TT$ with $M$
  nodes, there exists a tree $\TT'$ with at most $2N$ nodes, such that the
  depths of any two leaves in $\TT'$ differ by at most one, and a sequence of
  $\TT'$ nodes
  $V'=e'_1,\ldots, e'_N$, such that invoking $\reprupd(e', \mathit{root}(\TT'))$ for each
  $e'\in V'$ results in at least as many recursive $\reprupd$ calls as invoking
  $\reprupd(e, \mathit{root}(T))$ for each $e\in V$.
\end{lemma}

\begin{proof}[sketch]
  $V$ can be replaced by a sequence $V'$ containing only leaves, and $\TT$ by
  a balanced binary tree $\TT'$ with at most $2N$ nodes, without
  increasing the number of recursive $\reprupd$
  calls. \fullproofappendix \qed
\end{proof}



\subsection{Maintaining Implicit Suffix Nodes and Main Result}\label{sec-case-suffix}

To conclude our treatment, we disucss handling case~\ref{case-suffix} in
lemma~\ref{lem-cases}: finding the most recent match of a pattern that
corresponds to a point in $\ST$
with an implicitly
represented suffix on the same edge. Once such an implicit suffix node is
identified, the most recent pattern position is trivially obtained (the
position of the corresponding suffix). Furthermore, identifying implicit suffix nodes has a
known solution: Breslauer and Italiano~\cite{breslauer_sufext} describe
how Ukkonen's algorithm can be augmented with a stack of \emph{band
  trees}, whose nodes map top $\ST$ edges, by which implicit suffix nodes are
maintained for amortized constant-time access, under linear-time suffix
tree online construction. \detailsappendix

\label{sec-bresl-updpoint}The band stack scheme has one additional use in our
scheme: in each $\ST$ update operation, Breslauer and Italiano's algorithm pops
a number of bands from the stack, and keeps the node that is the endpoint of
the last popped edge. This node is the first explicit node on the active path,
and, equivalently, the edge is the maximum-$\ltdepth$ internal edge whose endpoint
corresponds an active suffix. This coincides with our definition of the
\emph{update edge} in section~\ref{sec-updstrat}. Thus, we obtain the current
update edge in constant time.

\begin{theorem}\label{th-allnlogn}
  A suffix tree with support for locating, in an input stream, the most recent
  longest match of an arbitrary pattern $P$ in $O(|P|)$ time, can be
  constructed online in time $O(N\log N)$ using $O(N)$ space, where $N$ is the
  current number of processed characters.
\end{theorem}

\begin{proof}[sketch]
  By lemma~\ref{lem-balleaves}, the number
  of $\reprupd$ calls is $O(N\log N)$, each of which takes
  constant time, using a data structure for constant-time lowest
  common ancestor queries~\cite{cole2005dynamic}. This bounds the
  time for maintenance under case~\ref{case-unique} in
  lemma~\ref{lem-cases} to $O(N\log N)$. In case~\ref{case-suffix}, we
  achieve~$O(N)$ time by the data structure of Breslauer and
  Italiano. \fullproofappendix \qed
\end{proof}

We assert that an adversarial input exists that results
in $\mathrm{\Omega}(N\log N)$ recursive calls, and hence this worst-case bound
is tight. \detailsappendix

\section{Sliding Window}\label{sec-slide}

A major advantage of online suffix tree construction is its applicability for a
\emph{sliding window}: indexing only the most recent part  (usually a fixed
length) of the input stream~\cite{FiGr,SufComp}. We note that
our augmentations of Ukkonen's algorithm can efficiently support most recent
match queries in a sliding window of size~$W$:


\begin{corollary}\label{cor-slide}
  A suffix tree with support for locating, among the most recent $W$ characters
  of an input stream, the most recent
  longest match of an arbitrary pattern $P$ in $O(|P|)$ time, can be
  constructed online in time $O(N\log W)$ using $O(W)$ space, where $N$ is the
  current number of processed characters.
\end{corollary}
 

\begin{proof}[sketch]
  The suffix tree is augmented for indexing a sliding
  window using $O(W)$
  space with maintained time bound~\cite{SufComp,FiGr}. Deletion from the data
  structure for ancestor queries takes $O(1)$
  time~\cite{dietz1987two}. Node deletion
  from band trees takes
  $O(1)$ time using
  \emph{pmerge}~\cite{westbrook1992fast}. Hence, a $O(N\log W$) term
  obtained analogously to lemma~\ref{lem-balleaves}
  dominates. \fullproofappendix \qed
\end{proof}

\section{An Optimization for the Lempel-Ziv Case}\label{sec-lz}

While our data structure supports arbitrary most-recent-match queries, some related
work has considered only the queries that arise in Lempel-Ziv
factorization, i.e., querying $\STi{i}$ only for the longest match of $t_i\cdots t_N$. The desire for finding
the most recent occurrence of each factor is motivated by an improved
compression rate in a subsequent entropy coding pass.


Ferragina, Nitto, and Venturini~\cite{ferragina1} gave an $O(N)$ time algorithm
for this case, which is not online, and hence cannot be
applied to a sliding window. Crochemore, Langiu, and
Mignosi~\cite{CrochemoreRightmost} presented an online $O(N)$ time suffix tree
data structure that, under additional assumptions, circumvents the problem by
replacing queries for most recent match with queries for matches with lowest
possible entropy-code length.
%
%
An interesting question is whether the time complexity of our method can be
improved if we restrict queries to those necessary for Lempel-Ziv
factorization. We now sketch an augmentation for this case.

As characters of one Lempel-Ziv factor
are incorporated into $\ST$, we need not invoke $\reprupd$ for the
update edge in each iteration. Instead, we push each update edge on a
stack. After the whole factor has been incorporated, we pop edges and invoke $\reprupd$ for the reverse sequence, updating edge $e$
only if it would have increased $\pos(e)$. In other words, we ignore
any updates superseded by later updates during the same sequence of edge
pops.
In experiments we noted drastic reduction in recursive calls, but
whether worst case asymptotic time is reduced is an open question.
(Extended discussion in
appendix\citeappendix.\!)



\section{Conclusion}\label{sec-concl}

We have presented an efficient online method of maintaining most recent match
information in a suffix tree, to support optimal-time queries. The
question wheth\-er the logarithmic factor in the time complexity  of our method can be improved
upon is, however, still open. Furthermore, precise characteristics of application to restricted inputs
or applications (e.g.\@ Lempel-Ziv factorization) is subject to future research, as is
the practicality of the result for, e.g., data
compression use.

\arxivonly{
  \clearpage
  \section*{Appendix}

This appendix presents proofs (and extended proofs) omitted from the
main text, as well as some extended discussions and details.

\setcounter{lemma}{0}
\setcounter{theorem}{0}
\setcounter{corollary}{0}

\subsection*{A.1\hspace{1em}Lemma 1--5 with Full Proofs}

\begin{lemma}
  Let $e$ be the edge that represents $P$, and let the string corresponding to
  $e$'s endpoint be $PA$,
    $|A|\ge 0$. Precisely one of the following holds:
  \begin{enumerate}
  \item The position of the most recent occurrence of $P$ is also
    the position of the most recent occurrence of $PA$.
  \item There exists a suffix $PB$, $|B|\ge 0$ such that $|B|<|A|$.
  \end{enumerate}
\end{lemma}

\begin{proof}
  Since there is no branching node between $P$ and $PA$, we know that for any
  substring $PC$, $C$ and $A$ must match in the first $\min\{|A|,|C|\}$
  characters. If $|C|\ge|A|$, any occurrence of $PC$ (including
  the most recent one) is an occurrence of $PA$, and
  case~\ref{case-unique} holds. If $|C|<|A|$, then, since
  there is no branching node between $PC$ and $PA$, $PC$ is a prefix of a
  string $B$ that corresponds to an implicit suffix node on $e$, and we have
  case~\ref{case-suffix}. Clearly, $PB$ occurs more recently than $PA$, since
  $PB$ is a suffix and $|PB|<|PA|$.\qed
\end{proof}

\begin{lemma}
  For an internal edge $e$, let $A$ be the string corresponding to $e$'s
  endpoint, and $t_{i-|A|}\cdots t_{i-1}$ the most recent occurrence of $A$ in
  $T$. Then $e$ has a descendent $f$ in $\LT$ whose endpoint corresponds to
  $BA$ for some string $B$, and
  $\pos(f)=i-|B|-|A|$.
\end{lemma}

\begin{proof}
  $A$'s most recent occurrence appeared in iteration $i$. The
  update edge in iteration $i$ must consequently be an edge $f$ whose endpoint
  is $BA$ for some $B$, and the iteration updates $\pos(f)$ to $i-|BA| =
  i-|B|-|A|$.

 We now show that $e$ is an ancestor of $f$ in $\LT$.
  Let $A=A_L a A_R$ such that $A_L a$ marks $e$. By the definition of $\LT$, any $BA_L a$ is represented by a descendent of $e$. Since
  there is no branching node between $A_L a$ and $A_L a A_R$, $A_L a$ is never
  followed by a string different from $A_R$, and hence neither is $B A_L
  a$. Consequently, $BA_L a$ and $BA$ are both represented by $f$.\qed
\end{proof}

\begin{lemma}
  Execution time of $\mrmfind(e)$, where $e$ represents a string $P$ can be
  bounded by $O(|P|)$.
\end{lemma}
\begin{proof}
  Let $Q$ be the string that marks $e$. We have $|Q|\le|P|$. Traversing the
  path from $e$ to the root via suffix links takes $|Q|$ steps,
  since following a suffix link implies navigating to the position of a shorter
  string. The ancestor query in step~\ref{st-mrmfind-isancestor} can be
  supported in $O(1)$ time~\cite{cole2005dynamic}, and all other operations in
  $\mrmfind$ are trivially constant-time.\qed
\end{proof}

\begin{lemma}
  Given a sequence $V=e_1,\ldots, e_N$ of nodes to be updated in a tree $\TT$ with $M$
  nodes, there exists a tree $\TT'$ with at most $2N$ nodes, such that the
  depths of any two leaves in $\TT'$ differ by at most one, and a sequence of
  $\TT'$ nodes
  $V'=e'_1,\ldots, e'_N$, such that invoking $\reprupd(e', \mathit{root}(\TT'))$ for each
  $e'\in V'$ results in at least as many recursive $\reprupd$ calls as invoking
  $\reprupd(e, \mathit{root}(T))$ for each $e\in V$.
\end{lemma}

\begin{proof}
  We start from $\TT$ and $V$, modifying them in a series of steps. The end
  result of all adjustments is $\TT'$ and $V'$.

  First, we make $V$ contain only leaves. Consider $e = e_j\in V$, and let $f$ and $g$ be defined in relation to
  $e$ according to $\reprupd$. Replace $e$ in $V$ with some leaf $e'$
  whose ancestor is $e$, with the following restriction: If $g=e\neq f$
  (case~c), $e'$ must be in a subtree of $g$ other than that which contains $f$.
  Reversely, if $g=f\neq e$ (case~d), $e'$ is in a subtree of $g$ other
  than that which contains $e$. In either case, if $g$ has only one
  child, we add $e'$ as a new leaf below $g$. Thereby, we ensure that the
  transformation does not reduce the number of recursive calls, and
  contributes at most $N$ nodes.

  Next, we show that the tree can be balanced, by the following
  argument. Consider $e$, $f$, and $g$ as defined in $\reprupd$. Given the
  precious transformation of $V$, we can assume that $e$ and $f$ are
  leaves. Recursion in $\reprupd$ progresses iff $e$ and $f$ are in different
  subtrees of $g$. Let $k$ be the number of children of $g$, $m_i$ the number
  of leaves in $g$'s $i$th subtree, and $m_g$ the total number of leaves in the
  subtree rooted at $g$. The number of possibilities for choosing $e$ and $f$
  is $\prod_{i=1}^k\binom{m_g}{m_i}$, which is maximized if the number of
  leaves is as evenly distributed as possible among the subtrees of $g$. We
  move nodes between subtrees to even out the number of leaves, without
  changing the number of leaves or internal nodes. Applying for
  all internal nodes yields a balanced tree, where the depth of leaves differ
  by at most one. We are not restricted in choosing nodes for the modified
  update sequence in any other way, and can make use of the full choice made
  possible by the restructuring in order to make sure that we do not reduce the
  number of recursive calls. Hence, for some sequence, the number of recursive
  calls is at least the same, which concludes the proof.\qed
\end{proof}

Lemma~\ref{lem-mrmupd-nlogn}, leading up to theorem~\ref{th-allnlogn},
is completely omitted from the main text, and presented only in this appendix.

\begin{lemma}\label{lem-mrmupd-nlogn}
  The time for maintaining $\repr(e)$ for each $e$ so as to maintain
  property~\ref{prop-repr} during construction of a suffix tree over a string
  of length $N$ can be bounded by $O(N\log N)$.
\end{lemma}
\begin{proof}
  By lemma~\ref{lem-balleaves}, the total number of recursive calls in
  invoking $\reprupd$ is proportional to the maximum for a balanced tree, whose
  height is $O(\log N)$, when the number of nodes is linear in $N$. The time
  for each recursive call is constant, when a data structure for constant-time
  lowest common ancestor queries is
  employed.~\cite{cole2005dynamic}. Consequently, total time is at most
  $O(N\log N).$\qed
\end{proof}

Note that locating the update edge, discussed
in section~\ref{sec-bresl-updpoint}, is not included in the time
accounted for by lemma~\ref{lem-mrmupd-nlogn}.

\subsection*{A.2\hspace{1em}Description of Band Trees and Full Theorem
  and Corollary}

Breslauer and Italiano~\cite{breslauer_sufext} describe
augmentations of Ukkonen's algorithm by which implicit suffix nodes can be
maintained for amortized constant-time access, while maintaining linear suffix
tree construction time. Implicit suffix nodes on external edges
has cyclicity properties that can be used for computing their positions
without any extra storage. Implicit nodes on internal edges are maintained
through the use of a stack of \emph{bands}, where a band is a tree whose nodes
map to $\ST$ edges with equal edge labels, and whose edges correspond to suffix
links (i.e., it is a part of $\LT$). Breslauer and Italiano show that the band
stack, and an implicit suffix node position for one representative of each
band, can be maintained in amortized $O(1)$ time per $\ST$ update iteration,
and support $O(1)$ time implicit-node queries.

\begin{theorem}
  A suffix tree with support for locating, in an input stream, the most recent
  longest match of an arbitrary pattern $P$ in $O(|P|)$ time, can be
  constructed online in time $O(N\log N)$ using $O(N)$ space, where $N$ is the
  current number of processed characters.
\end{theorem}

\begin{proof}
  For case~\ref{case-unique} in lemma~\ref{lem-cases}, query correctness and
  $|P|$ time bound under maintenance of property~\ref{prop-repr} for
  pos-updating the update edge at each iteration, follow from
  lemmas~\ref{lem-find-corr} and~\ref{lem-find-time}. The method for maintaining
  property~\ref{prop-repr} is given in section~\ref{sec-reprupd}, and its $O(N\log
  N)$ time bound given by
  lemma~\ref{lem-mrmupd-nlogn}. Locating the update
  edge for pos-updating takes constant time, using the described data structure
  of Breslauer and Italiano, which also provides $O(N)$ maintenance time for
  case~\ref{case-suffix} in lemma~\ref{lem-cases}. The space usage of all
  described data structures is bounded by $O(N)$.\qed
\end{proof}

\begin{corollary}
  A suffix tree with support for locating, among the most recent $W$ characters
  of an input stream, the most recent
  longest match of an arbitrary pattern $P$ in $O(|P|)$ time, can be
  constructed online in time $O(N\log W)$ using $O(W)$ space, where $N$ is the
  current number of processed characters.
\end{corollary}
 

\begin{proof}
  Our augmentations of the suffix tree in itself do not alter its structure,
  and consequently, existing techniques for augmenting the suffix tree algorithm for
  to index a sliding window in $O(1)$ amortized time, limiting space usage to
  $O(W)$~\cite{SufComp,FiGr} are directly applicable. The additional
  data structures are:
  \begin{itemize}
  \item The data structure for ancestor queries used in $\mrmfind$ in
    section~\ref{sec-mrmfind}. Deletions in $O(1)$ time are
    available~\cite{dietz1987two}, which can keep the space usage down to the
    $O(W)$ tree size.
  \item The \emph{band trees} kept on a stack in order to be able to find the
    update edge in constant time. Breslauer and
    Italiano~\cite{breslauer_sufext} do not discuss deleting nodes from the
    band trees, but we note that the data structures for dynamic nearest marked
    ancestors they use for achieving $O(1)$ amortized time operations do also
    support leaf deletions with the same time bound by means of a \emph{pmerge}
    operation~\cite{westbrook1992fast}, again allowing space usage to be
    asymptotically bounded by the $O(W)$ tree size.
  \end{itemize}

  Analogously to the proof of lemma~\ref{lem-balleaves}, the number of
  recursive $\reprupd$ calls is at most proportional to $N$ times the height of
  a perfectly balanced tree. For the sliding window suffix tree of size
  $O(W)$, this contributes a dominating term of $O(N\log W$) to the time
  complexity.\qed
\end{proof}

\subsection*{A.3\hspace{1em}Discussion of Worst Case for General
  Case and Lempel-Ziv Optimization}

Lemma~\ref{lem-mrmupd-nlogn} does not state that any input exists that results
in $\mathrm{\Omega}(N\log N)$ recursive calls, but such an adversarial input
\emph{does} exist, and hence our analysis is tight. We now give an
informal elaboration on the nature of an adversarial input. (Since our main results
do not depend on the lower bound, we do not provide a formal argument
to support the existance of this input.)

With specified parameter $d$,
an adversary can choose symbol $t_i$ considering the most recent
previous occurrence of a string $Aa$, where $A=t_{i-d}\cdots t_{i-1}$, and let
$t_i\neq a$. This produces a pair of edge updates that reaches depth $d$ in
$\LT$. The resulting adversarial string is a sequence with cycle
length $2^d$, and the number of times $\reprupd$ reaches recursion
depth $d$ approaches half of the iterations. With $N=c 2^d$ for constant $c$, $d$
is $\Theta(\log N)$.

$\log N$ is close to $N/2$. For $d=2$, one such
sequence has cycle $abaaabbb$; for $d=3$, the corresponding cycle is
$aaaabaabbababbbb$.

In experiments, we have
observed the worst case behavior for constructed adversarial inputs
only, and neither for naturally occurring data nor random inputs.

Furthermore, we note that optimization in
section~\ref{sec-lz} yields $O(N)$ time for the Lempel-Ziv special
case for the given adversarial input, as well as for any other
string exhibiting a cycle of constant length. 
%
%
We have not found any adversarial input that produces $\mathrm{\Omega}(N\log N)$
time in this case, and as noted in section~\ref{sec-lz}, it is an open
question whether it achieves total $o(N\log N)$ time.
}

\end{document}